\documentclass[leqno,12pt]{amsart}
\usepackage[colorlinks,pagebackref,hypertexnames=false]{hyperref}

\usepackage{amsmath,amssymb,mathrsfs,mathtools,verbatim} 
\usepackage[protrusion=true]{microtype} 
\usepackage[alphabetic,backrefs]{amsrefs}
\usepackage{ae,aecompl}
\usepackage[margin=1in]{geometry} 
\usepackage[onehalfspacing]{setspace}

\usepackage[english]{babel}

\usepackage{bookmark}

\numberwithin{equation}{section}
\usepackage{microtype}

\usepackage{fancybox,fancyhdr,graphics,epsfig}

\numberwithin{equation}{section}

\newcommand{\bE}{{\bf E}}

\newcommand{\bP}{{\bf P}}

\newcommand{\cT}{\mathcal{T}}

\renewcommand{\epsilon}{\varepsilon}

\def\<{\mathopen{}\left<}
\def\>{\right>\mathclose{}}
\def\({\mathopen{}\left(}
\def\){\right)\mathclose{}}

\newtheorem{theorem}{Theorem}

\newtheorem{example}{Example}

\newtheorem{proposition}{Proposition}
\newtheorem{remark}{Remark}

\numberwithin{equation}{section}

\author{Gon\c calo Oliveira}
\address{Universidade Federal Fluminense}
\email{galato97@gmail.com}

\title{Early epidemic spread, percolation and Covid-19}

\date{April 2020}

\begin{document}
	
	\begin{abstract}
		Human to human transmissible infectious diseases spread in a population using human interactions as its transmission vector. The early stages of such an outbreak can be modeled by a graph whose edges encode these interactions between individuals, the vertices. This article attempts to account for the case when each individual entails in different kinds of interactions which have therefore different probabilities of transmitting the disease. The majority of these results can be also stated in the language of percolation theory.
		
		The main contributions of the article are: (1) Extend to this setting some results which were previously known in the case when each individual has only one kind of interactions. (2) Find an explicit formula for the basic reproduction number $R_0$ which depends only on the probabilities of transmitting the disease along the different edges and the first two moments of the degree distributions of the associated graphs. (3) Motivated by the recent Covid-19 pandemic, we use the framework developed to compute the $R_0$ of a model disease spreading in populations whose trees and degree distributions are adjusted to several different countries. In this setting, we shall also compute the probability that the outbreak will not lead to an epidemic. In all cases we find such probability to be very low if no interventions are put in place.
	\end{abstract}

\maketitle

\tableofcontents




\section{Introduction}

The vector by which several infectious diseases propagate in a population is the human to human interaction. It is therefore natural to model their spread using such interactions as the ``basic mechanism''. In general, it gives rise to a dynamic process which evolves in time with its early stages being reasonably well approximated by considering the patient zero as the root (apex) of a tree whose branches encode the interactions through which the disease can propagate. In such epidemiological and percolation problems it is commonly assumed that each individual has an equal probability to transmit the disease to any of its contacts. However, this is an over-simplification of the actual situation as the same person can have several classes of interactions. Namely, it is conceivable that it is more likely that an infected individual will transmit the disease to someone with which it maintains a close familiar relation than to someone which it sporadically meets. In this setting we shall encode these different probabilities of transmitting the disease by distinct trees $\cT_1 , \ldots , \cT_n$ all with the same root which is the patient zero. To each tree $\cT_i$ we associated a probability $T_i$ of transmitting the disease along an interaction modeled by the corresponding tree. In the case when the $T_i$ are all the same for each vertex (individual) it is known that the basic reproduction number $R_0$ controls the possibility of almost surely avoiding an epidemic, see \cite{Brauer,1,2,3,4} and \cite{Callaway} which phrases these results in terms of the percolation interpretation. When there are different $T_i$ a similar framework can be used to prove the following result.

\begin{theorem}\label{thm:Main_1}
	Suppose the basic reproduction number $R_0<1$, then the outbreak will almost surely, i.e. with probability one, not lead to an epidemic. On the other hand, if $R_0>1$ there is still a nonzero probability $\bP \in (0,1)$ that the outbreak will be contained. 
\end{theorem} 

In section \ref{sec:Generating_Function} we shall develop the framework of multivariate generating functions on which this work will be based. Section \ref{sec:R_0} will show how to use this framework to effectively compute $R_0$ and finally we will prove the main abstract results in section \ref{sec:Results}. We will then exemplify the theory with a few examples. These are instructive in order to unravel an explicit formula for $R_0$ which does not depend on any ``abstract'' generating function. Such a formula is deduced in section \ref{sec:R_0_2} where we show the following result.

\begin{theorem}\label{thm:Main_2}
	For each $j \in \lbrace 1 , \ldots , n \rbrace$ let $\bE(k_j)$ denote the average degree of the tree $\cT_j$ and $\sigma(k_j)$ the standard deviation of the associated degree distribution. Suppose that a fraction $Q$ of all infected individuals is completely isolated and does not transmit the disease to anyone. Then, if the probability of transmitting the disease along an arm of the tree $\cT_j$ is $T_j \in [0,1]$, the basic reproduction number is
	\begin{align*}
	R_0  = (1-Q) \sum_{j=1}^n T_j \bE(k_j) \left( 1 -  \frac{1}{\sum_{l=1}^n \bE(k_l) } \right) + (1-Q) \sum_{j=1}^n T_j  \sigma(k_j) \frac{ \sigma(k_j) }{\sum_{l=1}^n \bE(k_l) },
	\end{align*}
	or
	\begin{align*}
	R_0  = (1-Q) \sum_{j=1}^n T_j \left( \bE(k_j) \left( 1 -  \frac{1}{\sum_{l=1}^n \bE(k_l) } \right) +  \frac{ \sigma(k_j)^2 }{\sum_{l=1}^n \bE(k_l) } \right).
	\end{align*}
\end{theorem}

A somewhat interesting feature of the previous formula is that it together with the transmission probabilities $\lbrace T_i \rbrace_{i=1}^n$, it depends solely on the first two moments of the degree distribution of the trees encoding the interactions.

As a final application of the theory, and motivated by the recent outbreak of Covid-19, we reserve the last section to do some specific country analysis. We have attempted to make the parameters of the theory be somewhat adequate to model the initial spread of Covid-19 but the results should be regarded as an ``academic'' toy example. A more robust analysis using our framework is possible, but would require detailed knowledge on the habits, family ties/interactions, attendace of public gatherings and other features of the analyzed populations, which are not uniform in each country. That last section computes the relevant values of $R_0$ for the countries considered and the probabilities that the outbreak will be contained. As we shall see, these are very small and in order to increase it we will investigate the effect of quarantining part of the infected individuals.

\subsection*{Acknowledgments} 
Gon\c{c}alo Oliveira is supported by Funda\c{c}\~ao Serrapilheira 1812-27395, by CNPq grants 428959/2018-0 and 307475/2018-2, and by FAPERJ through the grant Jovem Cientista do Nosso Estado E-26/202.793/2019.

\section{Generating functions}\label{sec:Generating_Function}

\subsection{From the generating vertex}

Let $\cT$ be a graph having a tree structure and whose edges are divided into $n$ groups. Each of groups gives rise to subgraphs $\cT_i \subset \cT$ having the same vertices and the edges of the corresponding group. We will also assume that $\cT_i \subset \cT$ has a tree structure. For each of these trees $\cT_i$, for $i=1,\ldots , n$, we shall denote by $\lbrace P_i(k) \rbrace_{k \in \mathbb{N}_0}$ the corresponding degree distribution, i.e. for a randomly chosen vertex its degree is $k$ with probability $P_i(k)$. Using these we can construct the corresponding generating functions
$$G_i(x) = \sum_{k=0}^{+\infty} P_i(k) x^k,$$
and the joint generating function 
$$G(x_1 , \ldots , x_n) = \sum_{k_1=0}^{+\infty} \ldots \sum_{k_n=0}^{+\infty} \prod_{i=1}^{n} P_i (k_i) x_i^{k_i} = \prod_{i=1}^{n} G_i (x_i). $$

\begin{remark}
	One can readily check that $G_i(0)=0$, $G_i(1)=1$ and $G_i(x)$ converges for $|x|\leq 1$. Furthermore, we compute that $G_i'(x)=\sum_{k=1}^{+\infty} kP_i(k) x^{k-1}$ and thus the average degree $\bE(k_i)$ of the tree $\cT_i$ is
	$$\bE(k_i)= G_i'(1).$$
	Similarly, we find that from $G(x,\ldots ,x)$ the total average degree is
	$$\bE(k)= \sum_{i=1}^n G_i'(1) = \frac{d}{dx} G(x , \ldots , x) \Big\vert_{x=1}.$$
\end{remark}
Each of these trees corresponds to different classes of contacts which have unequal probability of transmitting the disease. For instance, people that live in the same house are more likely to transmit the disease to each other than those which occasionally meet on public transport. Thus, to each tree, i.e., to each class of interactions, we associate a probability of transmission $T_i$ and assume with no loss of generality that $T_1 > \ldots > T_n$. Fix $m_1, \ldots , m_n$, then the probability that a randomly picked first infected individual transmits the disease to $m_i$ other individuals along the tree $\cT_i$ is
$$\sum_{k_1=m_1}^{+\infty} \ldots \sum_{k_n=m_n}^{+\infty} \prod_{i=1}^{n} P_i (k_i)  {k_i \choose m_i} T_i^{m_i}(1-T_i)^{k_i-m_i}. $$
Further suppose there is a probability $Q$ that an infected individual is detected and quarantined in complete isolation. Then, if $m_1+\ldots+m_n \neq 0$, the probability above must be multiplied by a factor of $1-Q$ which accounts for the possibility that it is not quarantined. At this point, it is convenient to define the generating function
\begin{equation}\label{eq:generating_Z}
\begin{aligned}
Z(x) & = Q + (1-Q) \sum_{k_1=0}^{+\infty} \ldots \sum_{k_n=0}^{+\infty} \prod_{i=1}^{n} P_i (k_i) \sum_{m_1=0}^{k_1} \ldots \sum_{m_n=0}^{k_n}  {k_i \choose m_i} (xT_i)^{m_i}(1-T_i)^{k_i-m_i} \\
& =  Q + (1-Q) \sum_{k_1=0}^{+\infty} \ldots \sum_{k_n=0}^{+\infty} \left( \prod_{i=1}^{n} P_i (k_i)  \right) \left( \prod_{i=1}^{n} (xT_i + 1-T_i)^{k_i} \right) \\
& =  Q + (1-Q) \sum_{k_1=0}^{+\infty} \ldots \sum_{k_n=0}^{+\infty}  \prod_{i=1}^{n} P_i (k_i)   (xT_i + 1-T_i)^{k_i} \\
& =  Q + (1-Q) G( 1-T_1(1-x) , \ldots , 1-T_n(1-x) ) .
\end{aligned}
\end{equation}

\subsection{By following a random infection}

Suppose we place ourselves at a vertex which is obtained from following a randomly chosen transmission. As an element of the tree $\cT_i$, the ramification of this vertex is the number of remaining edges emanating from it. The probability that such a vertex has ramification $k_1 , \ldots , k_n$ is therefore proportional to 
$$(k_1+1)P_1(k_1+1) \prod_{i \neq 1}P_i(k_i)  + (k_2+1)P_2(k_2+1)\prod_{i \neq 2}P_i(k_i)  + \ldots + (k_n+1)P_n(k_n+1) \prod_{i \neq n}P_i(k_i) , $$
with each term accounting from the probability of arriving at the chosen vertex via a given tree. Normalizing this we find that such probability $\tilde{P}(k_1 , \ldots , k_n)$ is obtained from the previous formula by dividing by 
$$\bE(k)=\bE(k_1) + \ldots + \bE(k_n)=G_1'(1) + \ldots + G_n'(1),$$ 
i.e. the average total ramification. We have thus concluded that
$$\tilde{P}(k_1 , \ldots , k_n) = \frac{ \sum_{l=1}^n(k_l+1)P_l(k_l+1) \prod_{i \neq l}P_i(k_i)  }{\bE(k)}.$$
Associated with this we define the generating function
\begin{equation}\label{eq:generating_tilde_G}
\tilde{G}(x_1, \ldots , x_n) : =  \sum_{k_1 , \ldots , k_n} \tilde{P}(k_1 , \ldots , k_n) x_1^{k_1} \ldots x_n^{k_n}.
\end{equation}
For future reference, it is convenient to have this written in terms of the simpler generating functions $G_i(x_i)$ for $i=1, \ldots , n$. For this, we insert the formula for $\tilde{P}(k_1 , \ldots , k_n)$ previously obtained. This yields
\begin{align*}
\tilde{G}(x_1, \ldots , x_n) & =  \sum_{k_1 , \ldots , k_n} \tilde{P}(k_1 , \ldots , k_n) x_1^{k_1} \ldots x_n^{k_n} \\
& = \frac{1}{\bE(k)}   \sum_{k_1 , \ldots , k_n} \left( \sum_{l=1}^n(k_l+1)P_l(k_l+1) \prod_{i \neq l}P_i(k_i) \right)  x_1^{k_1} \ldots x_n^{k_n} \\
& = \frac{1}{\bE(k)}   \sum_{k_1 , \ldots , k_n} \left( \sum_{l=1}^n(k_l+1)P_l(k_l+1) x_l^{k_l} \prod_{i \neq l}P_i(k_i) x_i^{k_i} \right)  \\
& = \frac{1}{\bE(k)} \left( \sum_{l=1}^n G_l'(x_l) \prod_{i \neq l}G_i(x_i) \right)  ,
\end{align*}
which may be written in the following more explicit form
\begin{equation}\label{eq:tilde_G_2}
\tilde{G}(x_1, \ldots , x_n) = \frac{\sum_{l=1}^n G_l'(x_l) \prod_{i \neq l}G_i(x_i)}{\sum_{l=1}^n G_l'(1) }. 
\end{equation}

\begin{remark}
	Consider of a sole tree $\cT_i$, the corresponding ramification distribution is
	$$\tilde{P}_i(k)=\frac{(k+1)P_i(k+1)}{\bE(k_i)} = \frac{(k+1)P_i(k+1)}{G_i'(1)}.$$ 
	which the individual generating function
	$$\tilde{G}_i(x)=\sum_{k=0}^{+\infty} \tilde{P}_i(k) x^k =  \frac{\sum_{k=0}^{+\infty} (k+1)P_i(k+1) x^k}{G_i'(1)}= \frac{G_i'(x)}{G_i'(1)}.$$
	Then, we have $G_i'(x)=G_i'(1) \tilde{G}_i(x)$ which upon inserting in equation \ref{eq:tilde_G_2} yields
	$$\tilde{G}(x_1, \ldots , x_n) = \frac{\sum_{l=1}^n \tilde{G}_l(x_l) G_l'(1) \prod_{i \neq l}G_i(x_i)}{\sum_{l=1}^n G_l'(1) }.$$
\end{remark}

Let $m_1, \ldots , m_n \in \mathbb{N}_0$ with $m_1 + \ldots + m_n >0$. Using the distribution for the ramification, we conclude that by following the contacts of the trees $\cT_1$ up to $\cT_n$ a randomly infected individual infects $m_1$ up to $m_n$ other ones is
$$ (1-Q) \sum_{k_1 , \ldots , k_n} \tilde{P}(k_1 , \ldots , k_n) \prod_{i=1}^n {k_i \choose m_i} T_i^{m_i} (1-T_i)^{k_i-m_i} , $$
if any of the $m_i$ is nonzero. Based on this, we define the generating function
\begin{equation}\label{eq:generating_tilde_Z}
\begin{aligned}
\tilde{Z}(x) & = Q + (1-Q)  \sum_{k_1, \ldots , k_n} \tilde{P}(k_1 , \ldots , k_n)  \prod_{i=1}^n  \sum_{m_i=0}^{k_i}  {k_i \choose m_i} (xT_i)^{m_i} (1-T_i)^{k_i-m_i} \\
& =  Q + (1-Q) \sum_{k_1, \ldots , k_n} \tilde{P}(k_1 , \ldots , k_n) \prod_{i=1}^n (xT_i+1-T_i)^{k_i} \\
& = Q + (1-Q)\tilde{G}(1+T_1(x-1) , \ldots , 1+ T_n(x-1)).
\end{aligned}
\end{equation}

\section{The basic reproduction number}\label{sec:R_0}

The basic reproduction number, usually denoted by $R_0$, is defined as the average number of individuals which are infected by each previously infected one. In our setting this can be immediately computed as follows. First, suppose we stand at a randomly chosen individual which has ramification $(k_1, \ldots, k_n)$. The average number of individuals it infects is
\begin{align*}
a_{k_1 \ldots k_n} & = (1-Q) \sum_{m_1=0}^{k_1} \ldots  \sum_{m_n=0}^{k_n} (m_1 + \ldots + m_n) \prod_{i=1}^n {k_i \choose m_i} (T_i)^{m_i} (1-T_i)^{k_i-m_i} \\
& = (1-Q) \sum_{l=1}^n \sum_{m_1=0}^{k_1} \ldots  \sum_{m_n=0}^{k_n} m_l \prod_{i=1}^n {k_i \choose m_i} (T_i)^{m_i} (1-T_i)^{k_i-m_i} \\
& = (1-Q) \sum_{i=1}^n \sum_{m_i=1}^{k_i}  m_i  {k_i \choose m_i} (T_i)^{m_i} (1-T_i)^{k_i-m_i},
\end{align*}
where we have used the binomial formula to deduce $\sum_{m_i=0}^n {k_i \choose m_i} (T_i)^{m_i} (1-T_i)^{k_i-m_i} =1 $. For the average vertex we must weight this with the ramification distribution, i.e.
\begin{align*}
R_0:= \sum_{k_1 \ldots k_n} \tilde{P}(k_1 , \ldots , k_n) a_{k_1 \ldots k_n}
\end{align*}
Having in mind the formula for ${k_i \choose m_i}$ we find
$$m_i  {k_i \choose m_i} = \frac{k_i!}{(k_i-m_i)!(m_i-1)!} = k_i  \frac{(k_i-1)!}{((k_i-1)-(m_i-1))!(m_i-1)!} = k_i {k_i-1 \choose m_i-1} .$$
Inserting into the above equation for $R_0$ and using again the binomial formula gives
\begin{align*}
R_0 & = (1-Q) \sum_{k_1 \ldots k_n} \tilde{P}(k_1 , \ldots , k_n)  \sum_{i=1}^n \sum_{m_i=1}^{k_i} k_i {k_i-1 \choose m_i-1}  (T_i)^{m_i} (1-T_i)^{k_i-m_i} \\
& = (1-Q) \sum_{k_1 \ldots k_n} \tilde{P}(k_1 , \ldots , k_n)  \sum_{i=1}^n k_i T_i \sum_{m_i=1}^{k_i}  {k_i-1 \choose m_i-1}  (T_i)^{m_i-1} (1-T_i)^{(k_i-1)-(m_i-1)} \\
& = (1-Q) \sum_{k_1 \ldots k_n} \tilde{P}(k_1 , \ldots , k_n)  \sum_{i=1}^n k_i T_i .
\end{align*}
Comparing this with the formula for $\tilde{Z}(x)$ we conclude the following result.

\begin{proposition}\label{prop:R_0}
	The basic reproduction number can be obtained from the generating function $\tilde{Z}(x)$ for the distribution of individuals infected by following a randomly chosen infected individual, via
	$$R_0=\tilde{Z}'(1),$$
	which may also be written as
	\begin{align*}
	R_0 & = (1-Q) \sum_{j=1}^n T_j \frac{\partial \tilde{G}}{\partial x_j} \Big\vert_{x_1= \ldots = x_n=1}.
	\end{align*}
\end{proposition}

\section{Containing an outbreak}\label{sec:Results}

In this section we shall compute the probability that the infection as it propagates eventually dies out. In the case of a unique tree, i.e. assuming all interactions have the same probability of transmitting the disease, such a computation have been carried out in \cite{Brauer}. See also \cites{1,2,3,4} for related results and \cite{Callaway} for the same setup in the context of the theory of percolation. 

Consider an individual which has been infected by another one, i.e. a vertex of the tree $\cT$ which is not its root and consider the probability that the infections generated by that vertex disappear within $N$ generations. Denote such a probability by $p_N$, then

\begin{align*}
\bP_N & = Q + (1-Q) \sum_{m_1 , \ldots , m_n} \bP_{N-1}^{m_1 + \ldots + m_n} \left( \sum_{k_1  \geq m_1 , \ldots , k_n \geq m_n}  \tilde{P}(k_1 , \ldots , k_n) \prod_{i=1}^n {k_i \choose m_i} T_i^{m_i} (1-T_i)^{k_i-m_i} \right) \\
& =  Q + (1-Q) \sum_{m_1 , \ldots , m_n}  \sum_{k_1  \geq m_1}  \ldots \sum_{k_n \geq m_n} \tilde{P}(k_1 , \ldots , k_n) \prod_{i=1}^n {k_i \choose m_i} \bP_{N-1}^{m_i}T_i^{m_i} (1-T_i)^{k_i-m_i} \\
& =  Q + (1-Q) \sum_{k_1 , \ldots , k_n}  \sum_{m_1 = 0}^{k_1}  \ldots \sum_{m_1 = 0}^{k_1}  \tilde{P}(k_1 , \ldots , k_n) \prod_{i=1}^n {k_i \choose m_i} (\bP_{N-1}T_i)^{m_i} (1-T_i)^{k_i-m_i} ,
\end{align*}
which upon comparing with the definition of $\tilde{Z}(x)$ in equation \ref{eq:generating_tilde_Z} can equally be read as
\begin{equation}\label{eq:p_N_from_P_N-1}
\bP_N = \tilde{Z}(\bP_{N-1}) .
\end{equation}
By construction we must have
$$\bP_0=Q + (1-Q) \sum_{k_1  \geq 0 , \ldots , k_n \geq 0}  \tilde{P}(k_1 , \ldots , k_n) \prod_{i=1}^n (1-T_i)^{k_i},$$
with equation \ref{eq:p_N_from_P_N-1} yielding all the following iterations. From inspection we find that $\bP_1:=\tilde{Z}(\bP_0)>\bP_0$ and as $\tilde{Z}$ is increasing we find assuming $ \bP_{N-1} > \bP_{N-2}$ that
$$\bP_{N}-\bP_{N-1}=\tilde{Z}(\bP_{N-1})-\tilde{G}(\bP_{N-2}) >0,$$
which inductively proves that the sequence $\lbrace \bP_N \rbrace_{N \in \mathbb{N}_0}$ is increasing. Hence, the number 
$$\bP_{\infty}= \lim_{N \to + \infty} \bP_N \in [0,1],$$
is well defined and encodes the probability that the infection starting from any such individual eventually dies out. We can then conclude the following.

\begin{proposition}\label{prop:Fixed_Point}
	The probability $\bP_{\infty}$ that the chain of infections generated from a randomly infected individual eventually disappears in a finite number of generations satisfies
	$$\bP_\infty = \tilde{Z}(\bP_\infty).$$
	Furthermore, there ate most two fixed points of $\tilde{Z}$ in $[0,1]$.
\end{proposition}
\begin{proof}
	The fact that $\bP_\infty$ satisfies $\bP_\infty = \tilde{Z}(\bP_\infty)$ follows immediately from the preceding discussion. Hence, the only remaining item to be shown is that there are at most two fixed points of $\tilde{Z}$ in the interval $[0,1]$. The fact that there is at least one is obvious as $\tilde{Z}(1)=1$. We must now show that there is at most one other.\\
	We argue by contradiction and assume there are at least two other different fixed points $\bP_*< \bP_{**}$ both in $(0,1)$.\footnote{We assume with no loss of generality that $\bP_*,\bP_{**}$ are positive as $\tilde{Z}(0)= \bP_0 \neq 0$ and so $0$ can never be a fixed point.} First, notice that both $\tilde{Z}'$, $\tilde{Z}''$ are positive in $(0,1)$. Secondly, consider the function $f(x)=x-\tilde{Z}(x)$ whose zeros correspond to the fixed points of $\tilde{Z}$, we have
	$$f(0)>0  , \ \ f(P_*)=0 , \ \ f(P_**)=0 , \ \  f(1)=0.$$
	Hence, by the intermediate value theorem there must be two critical points $c_* \in (\bP_*, \bP_{**})$ and  $c_{**} \in (\bP_{**},1)$ of $f$. As in $(0,1)$ we have $f''(x)=- \tilde{Z}''(x) <0$ we have that each of these must a maximum. Again, by the intermediate value theorem, between the two maxima must be a minimum $c_{***} \in (c_* , c_{**})$ contradicting $f''(c_{***})<0$.
\end{proof}

Placing ourselves at the tip of the tree which originated the infection chain, the so called patient zero, the probability that the infection eventually dies out is
\begin{align*}
\bP & = Q + (1-Q) \sum_{k_1=0}^{+\infty} \ldots \sum_{k_n=0}^{+\infty} \sum_{m_1=0}^{k_1} \ldots \sum_{m_n=0}^{k_n} \left( \prod_{i=1}^{n} P_i (k_i)  {k_i \choose m_i} T_i^{m_i}(1-T_i)^{k_i-m_i} \right) p_{\infty}^{m_1 + \ldots + m_n} \\
& = Q + (1-Q) \sum_{k_1=0}^{+\infty} \ldots \sum_{k_n=0}^{+\infty} \sum_{m_1=0}^{k_1} \ldots \sum_{m_n=0}^{k_n}  \prod_{i=1}^{n} P_i (k_i)  {k_i \choose m_i} (p_\infty T_i)^{m_i}(1-T_i)^{k_i-m_i}  \\
& = Z(\bP_\infty),
\end{align*}
where the last equality follows from comparison with the formula \ref{eq:generating_Z} for the generating function for infections starting from the patient zero.

\subsection{The case when $R_0<1$}

We shall now prove that when the basic reproduction number is smaller than one, the chain of infections will almost surely extinguish.

\begin{theorem}\label{thm:R_0<1}
	If $R_0<1$, then both $\bP$ and $\bP_{\infty}$ equal $1$. 
\end{theorem}
\begin{proof}
	By proposition \ref{prop:Fixed_Point} we know that $\bP_{\infty}$ it must be a fixed point of $\tilde{Z}$. Furthermore, we know that $\tilde{Z}(1)=1$ and so the statement follows if we can show that there is no other fixed point of $\tilde{Z}$ in the interval $[0,1]$. In that direction we shall show that under the hypothesis that $R_0<1$ the map 
	$$\tilde{Z}:[0,1] \to [0,1]$$ 
	is a contraction and so has a unique fixed point which must therefore the $1$.  This follows immediately from realizing that $\tilde{Z}''$ is nonnegative and so $\tilde{Z}'(x) \leq \tilde{Z}'(1)=R_0$ by \ref{prop:R_0}. Hence, for $x,y \in [0,1]$ we have
	$$|\tilde{Z}(x) - \tilde{Z}(y)| \leq \left( \sup_{t \in [0,1]}\tilde{Z}'(t) \right) |x-y|  \leq R_0 |x-y|,$$
	which shows that $\tilde{Z}$ is a contraction if $R_0<1$.\\
	Finally, the fact that also $\bP=1$ is then a consequence of $\bP=Z(\bP_\infty)=Z(1)=1$.
\end{proof}

\begin{remark}[The minimum required quarantined]\label{rem:Min_Quaranteen}
	At the beginning of an outbreak the question arises of what is the minimum number of infected individuals that must be detected and subsequently quarantined in order to contain the possible epidemic outbreak.
	
	If the disease is already well known, such as flu, measles or any other standard disease, not Covid-19, then its ``free'' basic reproduction number $R_0^{free}$ is known. Of course, this may depend on local conditions of where the outbreak takes place. By ``free'' we intend to emphasize that this is the basic reproduction number when the disease is free to propagate without taking in account any non-pharmaceutical intervention directed to slow its spread.\\
	Now, suppose an aggressive testing capacity can be put in place in order to detect those which have been infected. We would like to know the minimal fraction $Q$ of infected individuals which must be completely isolated so that the outbreak is almost surely controlled without having to take any other measures. The answer, as we shall now see is that $Q>1-\frac{1}{R_0^{free}}$.\\
	The generating functions $\tilde{Z}^{free}$ and $\tilde{Z}$ can both be written as in equation \ref{eq:generating_tilde_Z} with the exception that $Q=0$ for $\tilde{Z}^{free}$. Hence, $\tilde{Z}=Q + (1-Q) \tilde{Z}^{free}$ and by Proposition \ref{prop:R_0} we find
	$$R_0=\tilde{Z}'(1)=(1-Q) (\tilde{Z}^{free})'(1)=(1-Q)R_0^{free}.$$
	Then, the condition $R_0<1$ required to apply Theorem \ref{thm:R_0<1} turns into
	$$Q>1-\frac{1}{R_0^{free}},$$
	which corroborates the common intuition behind the basic reproduction number. For example, suppose there is an outbreak of disease for which each infected individual is expected to infect $3$ others if nothing is done to prevent it, i.e. $R_0^{free}$, then we expect that in order to cut the chain of transmission less than a third of the infections can be allowed to transmit the disease. Indeed, from the above computation, at least $Q>2/3$ of the whole infected must be detected and isolated.
\end{remark}

\subsection{The case when $R_0 > 1$}

Finally, in the case when $R_0 > 1$ we shall now prove that there is still a positive, but not certain, probability that the infection disappears.

\begin{theorem}\label{thm:R_0>1}
	If $R_0>1$, then $\bP_\infty \in (0,1)$ and $\bP= Z(\bP_\infty) \in (0,1)$.
\end{theorem}
\begin{proof}
	In this setup we consider the function $f(x)=x-\tilde{Z}(x)$ used in the proof of Proposition \ref{prop:Fixed_Point}. This satisfies $f(0)=-\tilde{Z}(0)<0$, $f(1)=0$ and by Proposition \ref{prop:R_0} $f'(1)=1-R_0<0$ and so $f(1-\epsilon)>0$ for sufficiently small nonzero $\epsilon \ll 1$. Thus, again the intermediate value theorem shows the existence of a zero of $f$ which we shall denote by $x_* \in (0,1-\epsilon)$. Recalling that zeros of $f$ correspond to fixed points of $\tilde{Z}$ which by Proposition \ref{prop:Fixed_Point} has only $x_*$ and $1$ as fixed points. Thus, in this case we can also have $\bP_\infty=x_*$.
\end{proof}

\begin{remark}[Lower bounds for $\bP$ and $\bP_\infty$]\label{rem:Lower_Bound}
	From the fixed point equation $\bP_\infty = \tilde{Z}(\bP_\infty)$
	and writing the generating function $\tilde{Z}$ as in Remark \ref{rem:Min_Quaranteen}, i.e. $\tilde{Z}(x)=Q+(1-Q) \tilde{Z}^{free}(x)$, we find that
	$$\bP_\infty=Q+(1-Q)\tilde{Z}^{free}(\bP_{\infty}),$$
	and so $\bP_\infty -Q = (1-Q)\tilde{Z}^{free}(\bP_{\infty}) >0$ and so
	$$\bP_\infty > Q.$$
	Furthermore, we can equally write $\bP=Z(\bP_\infty)=Q + (1-Q) \tilde{Z}(\bP_\infty)$. Said in other words, we find that the probability of the infection eventually dying out is at least the fraction of infected individuals which are completely isolated.
\end{remark}

\subsection{The case when $R_0=1$}

We shall now consider the case when $R_0=1$. We go back to the setup in the proof of Proposition \ref{prop:Fixed_Point} and Theorem \ref{thm:R_0<1}, namely we consider the function $f(x)=x-\tilde{Z}(x)$ whose zeros correspond to the fixed points of the map $\tilde{Z}:[0,1] \to [0,1]$. We have seen that $f''(x)<0$ in $(0,1)$ and $f(1)=0$. Under the hypothesis that $R_0=1$ we have $f'(1)=1-R_0=0$ and as $f''(x)<0$ for $x<1$ we find that $f$ is negative immediately before $x=1$. Hence, if there was another zero $x_* \in (0,1)$ of $f$, between $(x_*,1)$ the function $f$ would have a minimum which contradicts $f''<0$ in $(0,1)$. We then conclude that also in this case
$$\bP=\bP_{\infty}=1.$$

\subsection{Lower bounds on $\bP$}\label{sec:Lower_Bound}

In this section we shall elaborate on the question raised in Remark \ref{rem:Lower_Bound}, namely: Whether it is possible to find lower bounds on the probability that the chain of infections eventually dies out not leading to an epidemic. 

To answer the question raised we proceed by direct inspection of the fixed point equation in Proposition \ref{prop:Fixed_Point}. Start by noticing that all the terms in the Taylor series for $\tilde{Z}$ are positive as one can check from its definition in equation \ref{eq:generating_tilde_Z}. Thus, as $\bP_\infty=\tilde{Z}(\bP_\infty)$ we find that $\bP_\infty$ is larger than the zeroth order term of $\tilde{Z}$, i.e.
$$\bP_\infty > \tilde{Z}(0),$$
and from the monotonicity of $Z(x)$ we then have
$$\bP> Z(\tilde{Z}(0)),$$
which is itself grater than $Z(0)$.

\section{Examples}

In the simplest nontrivial example we can consider a population in which individuals have to kinds of interactions: a close and continuous interaction with their family and friends, and a a more distant sporadic interaction with not so close friends and other people which cross their path, by chance, in their daily lives as they commute to work and so on. Of course, we expect the probability of transmitting the disease to be larger in the first case and so assign to a it a larger transmissibility $T_1$ than to the second interactions $T_2$, i.e. $T_1>T_2$.

\subsection{Delta and Poisson}

In this first example we assume for simplicity that all person have the same number, $N_1 \geq 1$, of close contacts and their sporadic contacts follow a Poisson distribution with intensity $N_2$. In formulas, we have
$$P_1(k_1)=\delta_{N_1 k_1} , \ \text{and} \ P_2(k_2)=\frac{N_2^{k_2}}{k_2!} e^{-N_2}. $$
Then, we find the generating function for the degree distribution
$$G_1(x_1)=x_1^{N_1} , \ \text{and} \ G_2(x_2)= e^{-N_2(1-x_2)}, $$
from which we compute $G(x_1,x_2)=x_1^{N_1}e^{-N_2(1-x_2)}$ and so
\begin{align*}
Z(x) & = Q + (1-Q) G(1-T_1(1-x), 1-T_2(1-x)) \\
& = Q + (1-Q) (1-T_1(1-x))^{N_1} e^{-N_2 T_2(1-x)}.
\end{align*}
In order to compute $\tilde{Z}(x)$ we may first find the generating function $\tilde{G}(x_1 , x_2)$ for the ramification distribution. This can be done using equation \ref{eq:tilde_G_2} which yields
\begin{align*}
\tilde{G}(x_1,x_2) & = \frac{G_1(x_1)G_2'(x_2) + G_1'(x_1)G_2(x_2)}{G_1'(1)+G_2'(1)} \\
& = \frac{x_1^N N_2 e^{-N_2(1-x_2)} + N_1 x_1^{N_1-1} e^{-N_2 (1-x_2)}}{N_1 + N_2} \\
& = \frac{ N_1 + N_2 x_1 }{N_1 + N_2} x_1^{N_1-1}e^{-N_2(1-x_2)}.
\end{align*}
We can finally use this to compute the generating function $\tilde{Z}$ using the formula \ref{eq:generating_tilde_Z}. This gives
\begin{align*}
\tilde{Z}(x) & = Q + (1-Q) \tilde{G}(1-T_1(1-x), 1-T_2(1-x)) \\
& = Q + (1-Q) \frac{ N_1 + N_2 (1-T_1(1-x)) }{N_1 + N_2} (1-T_1(1-x))^{N_1-1}e^{-N_2(1-1+T_2(1-x))} \\
& = Q + (1-Q) \frac{ N_1 + N_2 - N_2 T_1(1-x) }{N_1 + N_2} (1-T_1(1-x))^{N-1}e^{-N_2 T_2(1-x)} \\
& = Q + (1-Q) \left( 1-  \frac{ N_2 }{N_2+N_1} T_1(1-x) \right) (1-T_1(1-x))^{N_1-1}e^{-N_2 T_2(1-x)} ,
\end{align*}
and using it we can compute the basic reproduction number $R_0$, which by Proposition \ref{prop:R_0} is
\begin{align*}
R_0 & = (1-Q) \frac{N_2 N_1 }{N_2 + N_1} \left( \frac{N_1-1}{N_2} T_1 + \frac{N_2}{N_1} T_2 + (T_1+T_2) \right),
\end{align*}
or perhaps in a somewhat more suggestive manner
\begin{align*}
R_0 & = (1-Q) \left( N_2 T_2 +  \left(N_1-\frac{N_1}{N_2 + N_1}\right) T_1 \right).
\end{align*}
Notice in particular that $R_0$ scales homogeneously with degree $1$ as a function of $N_1,N_2,T_1,T_2$.

\begin{remark}
	Notice that in the case $N_1=0$, i.e. if only the sporadic contacts occur, we have $R_0=(1-Q)N_2T_2$, while if $N_2=0$ meaning that all sporadic contacts are cut out, then $R_0=(1-Q)(N_1-1)T_1$.
\end{remark}

Also, by evaluating the generating function at $x=0$ we find
$$\tilde{Z}(0)=Q + (1-Q) \left( 1-  \frac{ N_2 }{N_2+N_1} T_1 \right) (1-T_1)^{N_1-1}e^{-N_2 T_2}.$$
It then  follows from the discussion in section \ref{sec:Lower_Bound} that
$$\bP_\infty > Q + (1-Q) \left( 1-  \frac{ N_2 }{N_2+N_1} T_1 \right) (1-T_1)^{N_1-1}e^{-N_2 T_2},$$
and 
$$\bP > Z(\tilde{Z}(0)) > Z(0) = Q + (1-Q)(1-T_1)^{N_1} e^{-N_2 T_2}.$$

\subsection{Polynomial and Poisson I}\label{ss:Example2}

In this second model we will elaborate slightly on the first example, in the sense that we still assume everyone to establishes sporadic contacts following a Poisson distribution with intensity $N_2$. On the other hand, we shall encode the distribution of close contacts by a polynomial of degree $N \geq 1$. Furthermore, let $N_1$ be the average degree of the distribution of close contacts, i.e. $N_1:=G_1'(1)$. In formulas, we have
$$P_1(k_1)= \sum_{i=1}^{N} \delta_{i k_1} p_{k_1} , \ \text{and} \ P_2(k_2)=\frac{N_2^{k_2}}{k_2!} e^{-N_2}. $$
As in the previous case, we can now find the generating function for these degree distributions
$$G_1(x_1)=p_0 + p_1 x + \ldots + p_{N_1}x_1^{N_1} , \ \text{and} \ G_2(x_2)= e^{-N_2(1-x_2)}, $$
and compute $G(x_1,x_2)=G_1(x_1)e^{-N_2(1-x_2)}$. Then,
\begin{align*}
Z(x) & = Q + (1-Q) G_1(1-T_1(1-x)) e^{-N_2 T_2(1-x)},
\end{align*}
and to compute $\tilde{Z}(x)$ we start by obtaining the generating function $\tilde{G}(x_1 , x_2)$ for the ramification distribution.Equation \ref{eq:tilde_G_2} yields
\begin{align*}
\tilde{G}(x_1,x_2) & = \frac{ (\log G_1(x_1))' + N_2  }{N_1 + N_2} G_1(x_1) e^{-N_2(1-x_2)}.
\end{align*}
Finally using equation \ref{eq:generating_tilde_Z} as before, we find
\begin{align*}
\tilde{Z}(x) & = Q + (1-Q) \frac{ G_1'(1-T_1(1-x)) + N_2 G_1(1-T_1(1-x)) }{N_1 + N_2} e^{-N_2 T_2(1-x)} ,
\end{align*}
from which we can compute $R_0$ using Proposition \ref{prop:R_0}, i.e. using the formula $\tilde{Z}'(1)$. This requires computing the derivative of $\tilde{Z}$ which reads
\begin{align*}
\tilde{Z}'(x) & =  (1-Q)T_1 \frac{ G_1''(1-T_1(1-x)) + N_2 G_1'(1-T_1(1-x)) }{N_1 + N_2} e^{-N_2 T_2(1-x)} \\
& \ \ + (1-Q) N_2 T_2 \frac{ G_1'(1-T_1(1-x)) + N_2 G_1(1-T_1(1-x)) }{N_1 + N_2} e^{-N_2 T_2(1-x)} ,
\end{align*}
and so
\begin{align*}
R_0 & =  (1-Q)T_1 \frac{ G_1''(1) + N_2 N_1 }{N_1 + N_2}  + (1-Q) N_2 T_2  .
\end{align*}

\begin{remark}\label{rem:Second_Moment}
In order to evaluate $R_0$ more clearly, we must understand what is $G_1''(1)$. This can be computed for any distribution $\lbrace P(k) \rbrace_{k \in \mathbb{N}_0}$ with generating function $G(x)=\sum_{k=0}^{+\infty} P(k)x^{k}$. Indeed, from direct differentiation
\begin{align*}
G''(1) & = \sum_{k=2}^{+\infty} k(k-1) P(k)  =  \sum_{k=0}^{+\infty} k^2 P(k)-  \sum_{k=0}^{+\infty} k P(k) = \bE(k^2)-\bE(k),
\end{align*}
as the $k=0$ terms vanish an the $k=1$ terms cancel.
\end{remark}

Inserting this into the previous formula for $R_0$ yields
\begin{align*}
R_0 & = (1-Q) \left( \frac{\bE(k_1^2) + N_1(N_2-1) }{N_2 + N_1} T_1 + N_2 T_2  \right).
\end{align*}
Notice in particular that $R_0$ scales homogeneously with degree $1$ as a function of $N_1,N_2,T_1,T_2$.

We turn now to compute lower bounds on $\bP$ and $\bP_\infty$ following the strategy of section \ref{sec:Lower_Bound}, from which we infer that
$$\bP_\infty > \tilde{Z}(0)=Q + (1-Q) \frac{ G_1'(1-T_1) + N_2 G_1(1-T_1) }{N_1 + N_2} e^{-N_2 T_2},$$ 
while
$$\bP > Z(\tilde{Z}(0)) > Z(0) = Q + (1-Q) G_1(1-T_1) e^{-N_2 T_2}.$$

\subsection{Polynomial and Poisson II}\label{ex:3}

We shall now continue with a population organized as in the previous example but we assume that a fraction $f$ of the population decides to social isolate and cut its sporadic contacts. One can imagine this can be done by not going into public transport, reducing contact with unknown people and working from home. To implement this we modify the previous example by writing
$$P_1(k_1)= \sum_{i=1}^{N} \delta_{i k_1} p_{k_1} , \ \text{and} \ P_2(k_2)= f \delta_{0k_2} + (1-f)\frac{N_2^{k_2}}{k_2!} e^{-N_2}. $$
As in the previous case, we can now find the generating function for these degree distributions
$$G_1(x_1)=p_0 + p_1 x + \ldots + p_{N_1}x_1^{N_1} , \ \text{and} \ G_2(x_2)= f + (1-f)e^{-N_2(1-x_2)}, $$
and compute $G(x_1,x_2)=G_1(x_1) (f + (1-f)e^{-N_2(1-x_2)})$. Then,
\begin{align*}
Z(x) & = Q + (1-Q) G_1(1-T_1(1-x)) (f + (1-f)e^{-N_2T_2(1-x)}),
\end{align*}
and
\begin{align*}
\tilde{G}(x_1,x_2) & = f \frac{ G_1'(x_1) }{N_1 + N_2} + (1-f) \frac{G_1'(x_1) + N_2 G_1(x_1)}{N_1+N_2} e^{-N_2(1-x_2)},
\end{align*}
from which we can compute
\begin{align*}
\tilde{Z}(x) & = Q + (1-Q) f \frac{ G_1'(1-T_1(1-x)) }{N_1 + N_2} \\
& \ \ + (1-Q) (1-f) \frac{G_1'(1-T_1(1-x)) + N_2 G_1(1-T_1(1-x))}{N_1+N_2} e^{-N_2T_2(1-x)} ,
\end{align*}
from which we can compute $R_0$ using Proposition \ref{prop:R_0}, i.e. using the formula $\tilde{Z}'(1)$. This requires computing the derivative of $\tilde{Z}$ which reads
\begin{align*}
\tilde{Z}'(x) & =  Q + (1-Q) f T_1 \frac{ G_1''(1-T_1(1-x)) }{N_1 + N_2} \\
& \ \ +  (1-Q) (1-f) T_1 \frac{G_1''(1-T_1(1-x)) + N_2 G_1'(1-T_1(1-x))}{N_1+N_2} e^{-N_2T_2(1-x)}  \\
& \ \ + (1-Q) (1-f) N_2 T_2 \frac{G_1'(1-T_1(1-x)) + N_2 G_1(1-T_1(1-x))}{N_1+N_2} e^{-N_2T_2(1-x)}
\end{align*}
and evaluating this at $x=1$ yields
\begin{align*}
R_0 & =  Q + (1-Q) f T_1 \frac{ G_1''(1) }{N_1 + N_2} \\
& \ \ +  (1-Q) (1-f) \left( T_1 \frac{G_1''(1) + N_2 G_1'(1)}{N_1+N_2}  + N_2 T_2 \frac{G_1'(1) + N_2 G_1(1)}{N_1+N_2} \right) \\
& =  Q + (1-Q) f T_1 \frac{ G_1''(1) }{N_1 + N_2}  +  (1-Q) (1-f) \left( T_1 \frac{G_1''(1) + N_1 N_2}{N_1+N_2} + N_2 T_2 \frac{N_1 + N_2}{N_1+N_2} \right).
\end{align*}
As before, if $R_0>1$, we find the lower bounds
$$\bP_\infty > \tilde{Z}(0) = Q + (1-Q) \left( f \frac{ G_1'(1-T_1) }{N_1 + N_2}  + (1-f) \frac{G_1'(1-T_1) + N_2 G_1(1-T_1)}{N_1+N_2} e^{-N_2T_2} \right),$$
and 
$$\bP>Z(\tilde{Z}(0))>Z(0)= Q + (1-Q) G_1(1-T_1) (f + (1-f)e^{-N_2T_2}).$$

\subsection{Poisson and Poisson}\label{ex:4}

In this final example we split the contacts again in two groups, the familiar close interactions and distant sporadic ones. In contrast to the previous examples we shall assume a Poisson distribution for both of these classes of contacts having intensity $N_1$ and $N_2$ respectively. As in example \ref{ex:3} we will be assuming that a fractions of the population are isolating by cutting their contacts. To work with some generality we will assume that fraction $f_1$ of the population cuts its close contacts and another fraction $f_2$\footnote{It is probably reasonable to assume that $f_2\geq f_1$.} cuts its sporadic contacts. Then, the degree distributions of the relevant trees $\cT_1$ and $\cT_2$ are
$$P_1(k_1)= f_1 \delta_{0k_1} + (1-f_1)\frac{N_1^{k_1}}{k_1!} e^{-N_1} , \ \text{and} \ P_2(k_2)= f_2 \delta_{0k_2} + (1-f_2)\frac{N_2^{k_2}}{k_2!} e^{-N_2}, $$
with the corresponding generating functions being
$$G_1(x_1)=f_1 + (1-f_1)e^{-N_1(1-x_1)} , \ \text{and} \ G_2(x_2)= f_2 + (1-f_2)e^{-N_2(1-x_2)}, $$
and compute $G(x_1,x_2)= (f_1 + (1-f_1)e^{-N_1(1-x_1)}) (f_2 + (1-f_2)e^{-N_2(1-x_2)})$. Then,
\begin{align*}
Z(x) & = Q + (1-Q) (f_1 + (1-f_1)e^{-N_1T_1(1-x)}) (f_2 + (1-f_2)e^{-N_2T_2(1-x)}),
\end{align*}
while
\begin{align*}
\tilde{Z}(x) & = Q + (1-Q) \sum_{i=1}^2 \frac{N_i(1-f_i) f_j e^{-N_i T_i (1-x)}}{N_1(1-f_1)+N_2(1-f_2)} \\
& \ \ + (1-Q) (1-f_1)(1-f_2) \frac{N_1+N_2}{N_1(1-f_1)+N_2(1-f_2)} e^{-N_1T_1(1-x)-N_2T_2(1-x)} ,
\end{align*}
where in the first sum $j \neq i$. We now compute $R_0$ using Proposition \ref{prop:R_0} which requires computing $\tilde{Z}'(1)$. This yields
\begin{align*}
R_0 & = (1-Q) \left( (1-f_1) N_1 T_1 + (1-f_2)N_2T_2 \right) \\
& \ \ + (1-Q) \frac{(1-f_1)f_1 N_1^2 T_1 + (1-f_2)f_2 N_2^2 T_2 }{(1-f_1)N_1 + (1-f_2) N_2}.
\end{align*}
If $R_0>1$, we have the lower bounds
\begin{align*}
\bP_\infty > \tilde{Z}(0) & =  Q + (1-Q) \sum_{i=1}^2 \frac{N_i(1-f_i) f_j e^{-N_i T_i }}{N_1(1-f_1)+N_2(1-f_2)} \\
& \ \ + (1-Q) (1-f_1)(1-f_2) \frac{N_1+N_2}{N_1(1-f_1)+N_2(1-f_2)} e^{-N_1T_1-N_2T_2} ,
\end{align*}
and 
$$\bP>Z(\tilde{Z}(0))>Z(0)= Q + (1-Q) (f_1 + (1-f_1)e^{-N_1T_1}) (f_2 + (1-f_2)e^{-N_2T_2}).$$

\section{An alternative formula for $R_0$}\label{sec:R_0_2}

This section is motivated by the examples in previous and in finding a more amenable general formula to compute $R_0$. We start with Proposition \ref{prop:R_0}, namely the equation 
$$
R_0  = (1-Q) \sum_{j=1}^n T_j \frac{\partial \tilde{G}}{\partial x_j} \Big\vert_{x_1= \ldots = x_n=1},
$$
and the equation \ref{eq:tilde_G_2} which we rewrite here for simplicity
$$
\tilde{G}(x_1, \ldots , x_n) = \frac{\sum_{l=1}^n G_l'(x_l) \prod_{i \neq l}G_i(x_i)}{\sum_{l=1}^n G_l'(1) }. 
$$
Using this, we compute
\begin{align*}
\frac{\partial \tilde{G}}{\partial x_j} \Big\vert_{x_1= \ldots = x_n=1} & =  \frac{G_j''(x_j) \prod_{i \neq j}G_i(x_i) + \sum_{l \neq j}^n G_l'(x_l)G_j'(x_j) \prod_{i \neq l,j}G_i(x_i)}{\sum_{l=1}^n G_l'(1) } \Big\vert_{x_1= \ldots = x_n=1} \\
& = \frac{G_j''(1) + G_j'(1) \sum_{l \neq j}^n G_l'(1) }{\sum_{l=1}^n G_l'(1) } 
\end{align*}
and from the discussion in Remark \ref{rem:Second_Moment} we further find
\begin{align*}
\frac{\partial \tilde{G}}{\partial x_j} \Big\vert_{x_1= \ldots = x_n=1}& = \frac{\bE(k_j^2)-\bE(k_j) + \bE(k_j) \sum_{l \neq k}^n \bE(k_l) }{\sum_{l=1}^n \bE(k_l) } \\
& = \bE(k_j)\frac{	\frac{\bE(k_j^2)}{\bE(k_j)} - 1 - \bE(k_j) + \sum_{l=1}^n \bE(k_l) }{\sum_{l=1}^n \bE(k_l) } \\
& = \bE(k_j)\frac{	\frac{\bE(k_j^2) - \bE(k_j)^2}{\bE(k_j)} - 1  + \sum_{l=1}^n \bE(k_l) }{\sum_{l=1}^n \bE(k_l) } \\
& = \frac{	\bE(k_j^2) - \bE(k_j)^2 }{\sum_{l=1}^n \bE(k_l) } + \bE(k_j) \left( 1 -  \frac{1}{\sum_{l=1}^n \bE(k_l) } \right) \\
	& = \frac{ \sigma(k_j)^2 }{\sum_{l=1}^n \bE(k_l) } + \bE(k_j) \left( 1 -  \frac{1}{\sum_{l=1}^n \bE(k_l) } \right),
\end{align*}
where we have used $\sigma(k_j)= \sqrt{\bE(k_j^2) - \bE(k_j)^2 }$ to denote the standard deviation of the degree distribution of the tree $\cT_j$. Then, inserting this into the formula for $R_0$ yields the formulae in Theorem \ref{thm:Main_2} which we shall restate here for convenience.

\begin{theorem}
	For each $j \in \lbrace 1 , \ldots , n \rbrace$ let $\bE(k_j)$ denote the mean degree of the tree $\cT_j$ and $\sigma(k_j)$ its standard deviation. Suppose that a fraction $Q$ of all infected individuals is completely isolated and does not transmit the disease to anyone. Then, if the probability of transmitting the disease along an arm of the tree $\cT_j$ is $T_j \in [0,1]$, the basic reproduction number is
	\begin{align*}
	R_0  = (1-Q) \sum_{j=1}^n T_j \bE(k_j) \left( 1 -  \frac{1}{\sum_{l=1}^n \bE(k_l) } \right) + (1-Q) \sum_{j=1}^n T_j  \sigma(k_j) \frac{ \sigma(k_j) }{\sum_{l=1}^n \bE(k_l) },
	\end{align*}
	or
	\begin{align*}
	R_0  = (1-Q) \sum_{j=1}^n T_j \left( \bE(k_j) \left( 1 -  \frac{1}{\sum_{l=1}^n \bE(k_l) } \right) +  \frac{ \sigma(k_j)^2 }{\sum_{l=1}^n \bE(k_l) } \right),
	\end{align*}
\end{theorem}

\section{Applications}

We shall now apply these results to a realistic scenario where a disease spreads through a population. Our goal is to investigate the possibility of an effective combination of isolation of infectious individuals, and practicing of social distancing, which together are capable of bringing $R_0$ below the threshold of $1$. When that is not possible we will compute the probability of an epidemic developing.

\subsection{Random contacts occurring with two different constant rates}

We shall use the setup of example \ref{ex:4} where the contacts established by the population in two groups. These are the close and distant contacts encoded in the trees $\cT_1$ and $\cT_2$ respectively. As in that example we assume both degree distributions to be Poisson, having intensities $N_1$ and $N_2$. The fractions of the population which are cutting their close contacts is $f_1$ and that cutting its sporadic contacts $f_2$. We will also be assuming that $f_2\geq f_1$ as it seems reasonable to assume that everyone which cuts its close contacts also cuts its sporadic ones. 

As computed in example \ref{ex:4}, the basic reproduction number is
\begin{align*}
R_0 & = (1-Q) \left( (1-f_1) N_1 T_1 + (1-f_2)N_2T_2 + \frac{(1-f_1)f_1 N_1^2 T_1 + (1-f_2)f_2 N_2^2 T_2 }{(1-f_1)N_1 + (1-f_2) N_2} \right).
\end{align*}
When no intervention is made all $f_1$, $f_2$ and $Q$ vanish the disease is free to propagate and the corresponding basic reproduction number will be denoted by
\begin{align*}
R_0^{free} & =  N_1 T_1 + N_2 T_2  .
\end{align*}

\begin{example}
	Suppose for the sake of simplicity that $f_1=f=f_2$. Then, $R_0$ can be reqritten in terms of $R_0^{free}$ as follows
	\begin{align*}
	R_0 & = (1-Q) \left( (1-f) (N_1 T_1 + N_2T_2) + f \frac{N_1^2 T_1 + N_2^2 T_2 }{N_1 + N_2} \right) \\
	& = (1-Q) \left( R_0^{free} - f \frac{N_1N_2}{N_1 + N_2} (T_1+T_2) \right).
	\end{align*}
	Then, the condition that $R_0<1$ which is sufficient to contain the outbreak turns into
	$$f>\frac{N_1+N_2}{N_1N_2} \frac{1}{T_1+T_2} \left( R_0^{free} - \frac{1}{1-Q} \right).$$
	For example, suppose that $Q=1/10$, $T_1=1/2$, $T_2=1/50$ and $N_1=4$, $N_2=50$. Then, $R_0^{free}=3$ which is actually a reasonable assumption a disease such as Covid-19 and the computation above yields $f>51/52$ which seems extremely difficult to achieve. On the other hand, if $Q=1/2$ while all other parameters remain the same $f>27/52 \approx 0.52$ which seems a much more achievable goal.
\end{example}

The previous example assumes that one can also cut the close contacts and that is not a reasonable assumption in most situations. For instance, many people may leave in the same house in which way it is not possible to cut such contacts. In the next example we address this and suppose only the sporadic contacts are cut.

\begin{example}
	Only cutting the sporadic contacts corresponds to having $f_1=0$. Then, $R_0$ can be is
	\begin{align*}
	R_0 & = (1-Q) \left( N_1 T_1 + (1-f_2)N_2T_2 + \frac{ (1-f_2)f_2 N_2^2 T_2 }{N_1 + (1-f_2) N_2} \right) \\
	& = (1-Q) \left( R_0^{free} -f_2N_2T_2 + \frac{ (1-f_2)f_2 N_2^2 T_2 }{N_1 + (1-f_2) N_2} \right) \\
	&  = (1-Q) \left( R_0^{free} -f_2 \frac{ N_1N_2T_2}{N_1 + (1-f_2) N_2} \right)
	\end{align*}
	and the condition that $R_0<1$ becomes
	$$f_2>\frac{N_1+N_2}{N_2}  \frac{ R_0^{free} - \frac{1}{1-Q} }{ R_0^{free} - \frac{1}{1-Q} +N_1 T_2}.$$
	We shall now use the same numerical values as in the previous example, $T_1=1/2$, $T_2=1/50$ and $N_1=4$, $N_2=50$. Then, we find from the previous computation that $f_2>1$ for any value of $Q\leq 1/2$. Hence, there is no way of surely containing the disease simply from cutting out the sporadic contacts and not increasing $Q$ above $1/2$. Suppose then the borderline case when $Q=1/2$ and consider the setting where also $f_2=1/2$ of the population is capable of cutting their sporadic contacts. Then, we have
	\begin{align*}
	Z(x) & = \frac{1}{2} + \frac{1}{4}e^{-2(1-x)} + \frac{1}{4} e^{-2502(1-x)} \\
	\tilde{Z}(x) & =\frac{1}{2} + \frac{1}{29}e^{-2(1-x)} + \frac{27}{58} e^{-2502(1-x)} ,
	\end{align*}
	from which we find
	$$\bP_\infty \approx 0.531 , \ \ \text{and} \ \ \bP \approx 0.594 ,$$
	i.e. the probability of an epidemic forming, which is given by $1-\bP$, is approximately $0.406$.
\end{example}

\subsection{Some country based analysis}

We shall now due some country based analysis using the household composition from the United Nations database \cite{UN}. This analysis is motivated by the current pandemic of Coronavirus. However, given the existence of insufficient knowledge to estimate the $T_i$ and the rude data from the beginning of the outbreak, this section should be regarded as within the realm of academic exercise. We shall use a simple model with two trees $\cT_1$ and $\cT_2$ with the first modeling the close interactions and given by the household composition of the respective country with the degree distribution being encoded in a generating function $G_1(x_1)$ which is then a polynomial. The second tree is intended to model the sporadic interactions in public gatherings. Assuming these to occur as a constant rate, the most suitable degree distribution is the Poisson of some intensity, say $N_2$, so that $G_2(x_2)=e^{-N_2(1-x_2)}$. Having this in mind, we shall use the model of sections \ref{ss:Example2} and \ref{ex:3} above. In all cases we shall consider, we have from those examples
$$G_1(x_1)=p_0 + p_1 x + \ldots + p_{5}x_1^{5} , \ \text{and} \ G_2(x_2)= f + (1-f) e^{-N_2(1-x_2)}, $$
so that writing $N_1=G_1'(1)$
\begin{align*}
Z(x) & = Q + (1-Q) G_1(1-T_1(1-x)) \left( f + (1-f)  e^{-N_2 T_2(1-x)} \right) \\
\tilde{Z}(x) & = Q + (1-Q) \frac{ G_1'(1-T_1(1-x))\left( f + (1-f)  e^{-N_2 T_2(1-x)} \right) + N_2 G_1(1-T_1(1-x)) e^{-N_2 T_2(1-x)} }{N_1 + N_2},
\end{align*}
and
\begin{align*}
R_0 & =  Q + (1-Q) f T_1 \frac{ G_1''(1) }{N_1 + N_2}  +  (1-Q) (1-f) \left( T_1 \frac{G_1''(1) + N_1 N_2}{N_1+N_2} + N_2 T_2 \frac{N_1 + N_2}{N_1+N_2} \right)  .
\end{align*}

\subsubsection{Germany}

We start with the case of Germany, which according to the United Nations database \cite{UN} has a household distribution in which $40 \% $ of the population live alone, $47 \%$ with one or two other individuals, $13\%$ with three or four other and $1\%$ with more than five other. The reader may note that we have not stated what is the exact percentage that live with only one or two other individuals. This is because such data is hard to locate for a large number of countries and we believe it will make little different in the qualitative, but also quantitative errors such as estimating the transmission probabilities $T_1$ and $T_2$. Indeed, there are more serious issues contributing to quantitative deviations. Thus, we will assume that half of the corresponding $47 \%$ live with one other individual and the other half with two other ones. Similar remarks hold for the $13\%$ of the population that lives with three or four other individuals. 

\begin{remark}
Of course, if all households remain isolated there is no way for the disease to spread from one household to another. However, we know this is not a realistic situation as in general there are close family ties connecting individuals in living in different houses. Given the difficulty in quantifying this we shall simply use the household composition to model the degree distribution of the tree $\cT_1$. 
\end{remark}

Based on the data mentioned above we shall assume that
$$G_1(x_1)=0.4 + 0.235 x_1 + 0.235 x_1^2+0.065 x_1^3+0.065 x_1^4+0.01 x_1^5.$$
As for the transmission rates, we will set $T_1=2/3$, i.e. the probability of infecting a close family member is $2/3$. In a similar way, we shall assume that the probability of infecting one of the sporadic contacts is $T_2=1/50$ and the average number of sporadic contacts is $N_2=100$. This includes everyone that an infected person stays next to in public transport, markets, restaurants, work and other common areas. Of course, these numbers are debatable and we have chosen these simply to illustrate the theory. Using them, and assuming that both $Q=0$ and $f=0$ we compute that $N_1=1.21$ and 
$$R_0 \approx 2.81 .$$
Then, iterating $\tilde{Z}(\bP_n)$ six times the sequence appears to stabilize around $\bP_\infty \approx 0.0852$ which gives
$$\bP \approx Z(\bP_\infty) \approx 0.0855 .$$
See figure \ref{fig:Germany} where the intersection point $\bP_\infty$ can be visualized graphically.
\begin{figure}[h]
	\centering
	\includegraphics[width=0.5\textwidth,height=0.3\textheight]{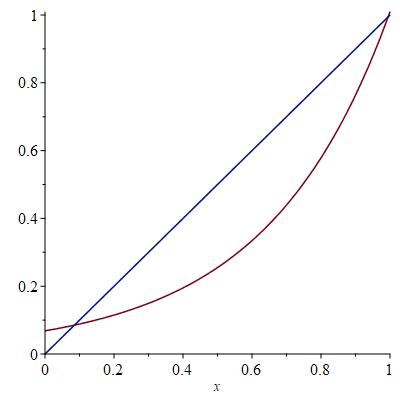}
	\caption{The function $\tilde{Z}(x)$ for Germany.}
	\label{fig:Germany}
\end{figure}
Hence, according to the model there is a very slight chance, of approximately $8.6\%$ that the outbreak will not lead to an epidemic. Of course, this assumes that no non-pharmaceutical interventions have been put in place to control the outbreak. Suppose for instance that strict social distancing outside the household is imposed so that everyone adheres to it, which corresponds to $f=1$. Then, a computation shows that $R_0 \approx 1.01 >1$ and so is not yet enough to almost surely guarantee that the outbreak will not lead to an epidemic. Also, getting all the population to cut its sporadic contacts seems very difficult to achieve in practice. A more efficient and easier to achieve strategy seems to be that of identifying and isolating infected individuals and we shall now analyze it. For example, motivated by the case of Covid-19 let us consider a disease for which around $17\%$ or $18\%$ of individuals are asymptomatic. These estimates where obtained for Covid-19 in airport screening and the data from the Diamond Princess cruise ship, see \cite{B}, \cite{Sun}. However, the validity of extrapolating thezse estimates to the remaining population is debatable as other studies seem to have quite disparate estimates for the fraction of asymptomatic carriers, see for instance \cite{M}. Let us say then, that at least $0.18$ of all infected individuals will not be detected, i.e $Q<1-0.18$. We choose, for simplicity $Q=0.7$ and $f=0$ which yields $R_0 \approx 0.84$ and the outbreak will almost surely be contained. As another example, suppose that $Q=0.6$ which is easier to achieve practically, then even with $f=0.9$ we have $R_0\approx 1.05 >1$ and so it would be needed $f>0.9$ which, again, is very difficult to achieve in practice. The conclusion is that, to contain the spread of the disease, it is much easier and effective to quarantine a sufficient fraction of infected individuals effectively than to simply cut the sporadic contacts of a large fraction of the entire population.

\subsubsection{Italy}

We shall now consider the Italian case, and based on \cite{UN} we shall assume that
$$G_1(x_1)=0.31+ 0.235 x_1 + 0.235 x_1^2 + 0.105x_1^3 + 0.105x_1^4 + 0.1x_1^5. $$
Then, using the same parameters as in the previous example we compute $R_0 \approx 3$ and $\bP_\infty \approx 0.068$ so that
$$\bP \approx Z(\bP_\infty) \approx 0.068,$$
i.e. there is chance of $6.8 \%$ that an outbreak can be avoided with taking any precautions. Still, as a matter of comparison we find that if $Q=0.7$ then $R_0 \approx 0.9$ which even though below $1$ is visibly higher than that of the previous example.

\subsubsection{France}

For modeling France, we set
$$G_1(x_1)=0.35 + 0.235 x_1 + 0.235 x_1^2 + 0.08 x_1^3 + 0.08 x_1^4 + 0.2 x_1^5. $$
which with the same parameters of the previous two examples yields $R_0 \approx 3.09$ and $\bP_\infty \approx 0.039$ which gives
$$\bP \approx Z(\bP_\infty) \approx 0.066,$$
yielding a probability of $6.6 \%$ to avoid an outbreak. When $Q=0.7$ we compute $R_0 \approx 0.92 $ which being below $1$ is larger than that of the previous two examples.

\subsubsection{Portugal}

For Portugal, the United Nations database \cite{UN} suggests using 
$$G_1(x_1)=0.19 + 0.28 x_1 + 0.28 x_1^2 + 0.115 x_1^3 + 0.115 x_1^4 + 0.2 x_1^5. $$
Again, with the same parameters used in the previous examples we find $R_0 \approx 3.4$ and $\bP_\infty \approx 0.04$. This gives
$$\bP \approx Z(\bP_\infty) \approx 0.04,$$
which yields a probability of $3.9 \%$ to avoid an outbreak. In this case, when $Q=0.7$ we find $R_0 \approx 1.019 $ which in contrast with the previous examples is already above $1$. Thus, by Theorem \ref{thm:R_0<1} the probability $\bP$ of avoiding an epidemic is below $1$. Nevertheless, a computation shows that $\bP \approx 0.984$ which is still quite high.

\subsubsection{Spain}

In the case of Spain we assume, using the same reference, that
$$G_1(x_1)=0.19 + 0.265 x_1 + 0.265 x_1^2 + 0.13 x_1^3 + 0.13 x_1^4 + 0.3 x_1^5. $$
Again, with the same parameters used in the previous examples we find $R_0 \approx 3.44$ and $\bP_\infty \approx 0.039$. This gives
$$\bP \approx Z(\bP_\infty) \approx 0.039,$$
which yields a probability of $3.9 \%$ to avoid an outbreak. In this case, when $Q=0.7$ we find $R_0 \approx 1.03 $ as in the previous example. Indeed, the whole situation is very much parallel as that of the previous example.

\subsubsection{Brazil}

For Brazil we have
$$G_1(x_1)=0.12 + 0.235 x_1 + 0.235 x_1^2 + 0.16 x_1^3 + 0.16 x_1^4 + 0.9 x_1^5, $$
and with the same parameters of the previous examples we find $R_0 \approx 3.82$ which is the highest so far. Using this distributions we compute $\bP_\infty \approx 0.026$ which has
$$\bP \approx Z(\bP_\infty) \approx 0.026,$$
which yields a very small probability of $2.6 \%$ to avoid an outbreak. In this case, even when $Q=0.7$ we find $R_0 \approx 1.146 $ which is still high. Indeed, the associated probability of avoiding an outbreak is $\bP \approx 0.92$, i.e. there is a chance of $92\%$ of containing the outbreak.

%
%
%

\end{document}